\documentclass[a4,11pt]{article}

\usepackage{microtype} 
\usepackage{tikz}
\usepackage{calc}
\usepackage{amsfonts}
\usepackage{amsthm}
\usepackage{amssymb}
\usepackage{mathtools}
\usepackage{subcaption}
\usepackage{caption}
\usepackage{url}

\usetikzlibrary{shapes,positioning,chains,fit,calc,decorations.markings,decorations.pathreplacing}
\tikzstyle{ccyan}=[circle, draw, thick,fill=cyan!30, minimum size=12pt,inner sep=0pt]
\tikzstyle{cgrey}=[circle, draw, thick,fill=gray!30, minimum size=10pt,inner sep=0pt]
\tikzstyle{cgreys}=[circle, draw, thick,fill=gray!30, minimum size=12pt,inner sep=0pt]

\newcommand{\ket}[1]{| #1 \rangle}
\newcommand{\bra}[1]{\langle #1 |}
\newcommand{\braket}[2]{\langle #1 | #2 \rangle}

\def\U{\uparrow}
\def\D{\downarrow}
\def\LL{\leftarrow}
\def\R{\rightarrow}

\newcommand{\comment}[1]{}

\newtheorem{theorem}{Theorem}
\newtheorem{lemma}{Lemma}
\newtheorem{corollary}{Corollary}

\title{Adjacent vertices can be hard to find by quantum walks\footnote{This work was supported by the European Union Seventh Framework
Programme (FP7/2007-2013) under the QALGO (Grant Agreement No.~600700)
project and the RAQUEL (Grant Agreement No.~323970) project, the Latvian State Research Programme NeXIT
project No.~1, and the ERC Advanced Grant MQC.}}

\author{Nikolajs Nahimovs and Raqueline A. M. Santos}
\date{\small{Faculty of Computing, University of Latvia} \\ 
\small{Raina bulv. 19, Riga, LV-1586, Latvia}\\
\small{\texttt{nikolajs.nahimovs@lu.lv, rsantos@lu.lv}}}

\begin{document}

\maketitle


\begin{abstract}

\noindent
Quantum walks have been useful for designing quantum algorithms that outperform their classical versions for a variety of search problems.
Most of the papers, however, consider a search space containing a single marked element only.
We show that if the search space contains more than one marked element, their placement may drastically affect the performance of the search. 
More specifically, we study search by quantum walks on general graphs and show a wide class of configurations of marked vertices, for which search by quantum walk needs $\Omega(N)$ steps, that is, it has no speed-up over the classical exhaustive search.
The demonstrated configurations occur for certain placements of two or more adjacent marked vertices. The analysis is done for the two-dimensional grid and hypercube, and then is generalized for any graph. 

\end{abstract}

 
\section{Introduction}

Quantum walks are quantum counterparts of classical random walks \cite{Portugal:2013}. 
Similarly to classical random walks, there are two types of quantum walks: discrete-time quantum walks, first introduced by Aharonov~{\it et al.}~\cite{Aharonov:1993}, and continuous-time quantum walks, introduced by Farhi~{\it et al.}~\cite{Farhi:1998}.
For the discrete-time version, the step of the quantum walk is usually given by coin and shift operators, which are applied repeatedly. 
The coin operator acts on the internal state of the walker and rearranges the amplitudes of going to adjacent vertices. The shift operator moves the walker between the adjacent vertices.

Quantum walks have been useful for designing algorithms for a variety of search problems \cite{Ambainis:2005,Magniez:2005,Ambainis:2004}.
To solve a search problem using quantum walks, we introduce the notion of marked elements (vertices), corresponding to elements of the search space that we want to find.
We perform a quantum walk on the search space with one transition rule at the unmarked vertices, and another transition rule at the marked vertices. If this process is set up properly, it leads to a quantum state in which marked vertices have higher probability than the unmarked ones. This method of search using quantum walks was first introduced in \cite{Shenvi:2003}, which describes a quantum search in the hypercube, and has been used many times since then.

%
Not many papers in the literature consider search by quantum walks with multiple marked vertices.  
Wong~\cite{Wong:2016} analyzed the spatial search problem solved by continuous-time quantum walk on the simplex of complete graphs and showed that the location of marked vertices can dramatically influence the required jumping rate of the quantum walk. Wong and Ambainis~\cite{Wong:2015} analysed the discrete-time quantum walk on the simplex of complete graphs and showed that if one of the complete graphs is fully marked then there is no speed-up over classical exhaustive search.
Nahimovs and Rivosh~\cite{Nahimovs:2015a} studied the dependence of the running time  of the AKR algorithm~\cite{Ambainis:2005} on the number and the placement of marked locations. 
They found some ``exceptional configurations'' of marked vertices, for which the probability of finding any of the marked vertices does not grow over time. Another previously known exceptional configuration for the two-dimensional grid is the ``diagonal construction'' by Ambainis and Rivosh~\cite{Ambainis:2008}.

In this paper, we extend the work of Nahimovs and Rivosh~\cite{Nahimovs:2015}. 
We study search by quantum walks on general graphs with multiple marked vertices and show a wide class of configurations of marked vertices, for which the probability of finding any of the marked vertices does not grow over time. 
These configurations occur for certain placements of two and more adjacent marked vertices.
We prove that for such configurations the state of the algorithm is close to a stationary state.

We start by reviewing the simple example of the two-dimensional grid from~\cite{Nahimovs:2015} by showing that any pair of adjacent marked vertices forms an exceptional configuration.
The same construction is valid for the hypercube.
We extend the proof to general graphs by showing that any pair of adjacent marked vertices having the same degree $d$ forms an exceptional configuration, for which the probability of finding a marked vertex is limited by ${\textrm{const}} \cdot d^2/N$.
Then, we prove that any $k$-clique of marked vertices forms an exceptional configuration.
Additionally, we formulate general conditions for a state to be stationary given a configuration of marked vertices.
Our results greatly extend the class of known exceptional configurations.

\section{Two-dimensional grid}\label{sec:grid}

\subsection{Quantum walk on the two-dimensional grid}

Consider a two-dimensional grid of size $\sqrt{N}\times\sqrt{N}$ with periodic (torus-like) boundary conditions. Let us denote $n=\sqrt{N}$. 
The locations of the grid are labeled by the coordinates $(x,y)$ for $x, y \in \{0,\dots,n-1\}$.
The coordinates define a set of state vectors, $\ket{x,y}$, which span the Hilbert space, ${\cal{H_P}}$, associated to the position. 
Additionally, we define a 4-dimensional Hilbert space with the set of states $\{\ket{c}: c\in \{\LL,\R,\U,\D \}\}$, ${\cal{H_C}}$, associated with the direction. We refer to it as the direction or the coin subspace. The quantum walk has associated Hilbert space ${\cal{H_P}}\otimes{\cal{H_C}}$ with basis states $\ket{x,y,c}$ for $x,y \in \{0,\dots,n-1\}$ and $c \in \{\U,\D,\LL,\R\}$.

The evolution of a state of the walk is driven by the unitary operator $U = S\cdot (I\otimes C)$, where $S$ is the flip-flop shift operator
\begin{eqnarray}
S\ket{i,j,\U} & = & \ket{i,j+1,\D} \\
S\ket{i,j,\D} & = & \ket{i,j-1,\U} \\
S\ket{i,j,\LL} & = & \ket{i-1,j,\R} \\
S\ket{i,j,\R} & = & \ket{i+1,j,\LL},
\end{eqnarray}
and $C$ is the coin operator, given by the Grover's diffusion transformation 
\begin{equation}
C = \frac{1}{2} \left( 
\begin{array}{cccc}
-1 & 1 & 1 & 1 \\
1 & -1 & 1 & 1 \\
1 & 1 & -1 & 1 \\
1 & 1 & 1 & -1 
\end{array} \right).
\end{equation}

The spatial search algorithm uses the unitary operator $U' = S\cdot (I\otimes C)\cdot(Q\otimes I)$, where Q is the query transformation which flips the sign of marked vertices, that is, $Q\ket{x,y} = -\ket{x,y}$, if $(x,y)$ is marked and $Q\ket{x,y} = \ket{x,y}$, otherwise. 
The initial state of the algorithm is
\begin{equation}\label{eq:psi0_grid}
\ket{\psi(0)} = \frac{1}{\sqrt{4N}} \sum_{i,j=0}^{n-1} \big( \ket{i,j,\U} + \ket{i,j,\D} + \ket{i,j,\LL} + \ket{i,j,\R} \big).
\end{equation}
Note that $\ket{\psi(0)}$ is a 1-eigenvector of $U$ but not of $U'$. 
If there are marked vertices, the state of the algorithm starts to deviate from $\ket{\psi(0)}$.
In case of one marked vertex, after $O(\sqrt{N\log{N}})$ steps the inner product $\braket{\psi(t)}{\psi(0)}$ becomes close to $0$.
If we measure the state at this moment, we will find the marked vertex with $O(1 / \log{N})$ probability~\cite{Ambainis:2005}.
This gives the total running time of $O(\sqrt{N} \log{N})$ steps with amplitude amplification.

By analyzing the quantum search algorithm for a group of marked vertices of size $\sqrt{k}\times \sqrt{k}$, Ref.~\cite{Nahimovs:2015} identified that the algorithm does not work as expected when $k$ is even, meaning that the overlap of the initial state with the state at time $t$, $\braket{\psi(0)}{\psi(t)}$, stays close to 1. Moreover, the same effect holds for any block of size $2k\times l$ or $k\times 2l$, with $l$ and $k$ being positive integers. The reason for such behavior is that blocks of marked vertices form stationary states, as we are going to see below.

\subsection{Stationary states for the two-dimensional grid}

Consider a two-dimensional grid with two marked vertices $(i,j)$ and $(i,j+1)$.
Let $\ket{\phi_{stat}^a}$ be a state having amplitudes of all basis states equal to $a$ except for $\ket{i,j,\R}$ and $\ket{i+1,j,\LL}$, which have amplitudes equal to $-3a$ (see Fig.~\ref{fig:1x2_stationary_state}), that is,

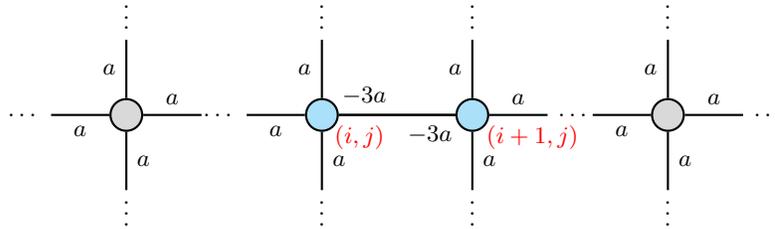
\begin{figure}[ht]
\centering
\begin{tikzpicture}[auto, thick]

\node[ccyan] (a1) at (-1,0) {} ;
\node[ccyan] (a2) at (1,0) {};
\node[red] at (1.8,-0.3) {\footnotesize{$(i+1,j)$}};
\node[red,thick,xshift=-1.8cm] at (1.3,-0.3) {\footnotesize{$(i,j)$}};
 
\draw[thick] (a1) -- node[above left]{\footnotesize{$-3a$}} (a2) ;
\draw[thick] (a1) -- node[below right]{\footnotesize{$-3a$}} (a2) ;
\draw[thick] (a2) -- node {\footnotesize{$a$}} (1,1);
\draw[thick] (a2) -- node {\footnotesize{$a$}} (2,0);
\draw[thick] (a2) -- node {\footnotesize{$a$}} (1,-1);
\draw[thick] (a1) -- node {\footnotesize{$a$}} (-1,1);
\draw[thick] (a1) -- node {\footnotesize{$a$}} (-2,0);
\draw[thick] (a1) -- node {\footnotesize{$a$}} (-1,-1);

\node[above] at (1,1) {\footnotesize{$\vdots$}};
\node[above] at (-1,1) {\footnotesize{$\vdots$}};
\node[below] at (1,-0.8) {\footnotesize{$\vdots$}};
\node[below] at (-1,-0.8) {\footnotesize{$\vdots$}};
\node[right] at (2,0) {\footnotesize{$\dots$}};
\node[left] at (-2,0) {\footnotesize{$\dots$}};

\begin{scope}[xshift=3.6cm]
\node[cgreys] (a3) at (0,0) {};
\draw[thick] (a3) -- node {\footnotesize{$a$}} (0,1);
\draw[thick] (a3) -- node {\footnotesize{$a$}} (1,0);
\draw[thick] (a3) -- node {\footnotesize{$a$}} (-1,0);
\draw[thick] (a3) -- node {\footnotesize{$a$}} (0,-1);
\node[below] at (0,-0.8) {\footnotesize{$\vdots$}};
\node[above] at (0,1) {\footnotesize{$\vdots$}};
\node[right] at (1,0) {\footnotesize{$\dots$}};
\end{scope}

\begin{scope}[xshift=-3.6cm]
\node[cgreys] (a3) at (0,0) {};
\draw[thick] (a3) -- node {\footnotesize{$a$}} (0,1);
\draw[thick] (a3) -- node {\footnotesize{$a$}} (1,0);
\draw[thick] (a3) -- node {\footnotesize{$a$}} (-1,0);
\draw[thick] (a3) -- node {\footnotesize{$a$}} (0,-1);
\node[below] at (0,-0.8) {\footnotesize{$\vdots$}};
\node[above] at (0,1) {\footnotesize{$\vdots$}};
\node[left] at (-1,0) {\footnotesize{$\dots$}};
\end{scope}
\end{tikzpicture}
\caption{The amplitudes of the state $\ket{\phi_{stat}^a}$. The vertices $(i,j)$ and $(i+1,j)$ are marked.}
\label{fig:1x2_stationary_state}
\end{figure}

\begin{equation}\label{eq:phi_grid}
\ket{\phi_{stat}^a} = \sum_{x,y=0}^{n-1}\sum_c a\ket{x,y,c} - 4a\ket{i,j,\R} - 4a\ket{i+1,j,\LL}.
\end{equation}
Then, this state is not changed by a step of the algorithm.

\begin{lemma}\label{lemma:grid}
Consider a grid of size $\sqrt{N} \times \sqrt{N}$ with two adjacent marked vertices $(i,j)$ and $(i+1,j)$. Then the state $\ket{\phi_{stat}^a}$, given by Eq.~(\ref{eq:phi_grid}), is not changed by the step of the algorithm, that is, $U'\ket{\phi_{stat}^a} = \ket{\phi_{stat}^a}$.
\end{lemma}
\begin{proof}
Consider the effect of a step of the algorithm on $\ket{\phi_{stat}^a}$. The query transformation changes the signs of all the amplitudes of the marked vertices. The coin transformation performs an inversion about the average: for unmarked vertices, it does nothing, as all amplitudes are equal to $a$; for marked vertices, the average is $0$, so applying the coin results in sign flip. Thus, $(I\otimes C)(Q\otimes I)$ does nothing for the amplitudes of the non-marked vertices and twice flips the sign of the amplitudes of  the marked vertices. Therefore, we have $(I\otimes C)(Q\otimes I)\ket{\phi_{stat}^a} = \ket{\phi_{stat}^a}.$
The shift transformation swaps the amplitudes of near-by vertices. For $\ket{\phi_{stat}^a}$, it swaps $a$ with $a$ and $-3a$ with $-3a$. Thus, we have $S(I\otimes C)(Q\otimes I)\ket{\phi_{stat}^a} = \ket{\phi_{stat}^a}$.
\end{proof}

The initial state of the algorithm, given by Eq.~(\ref{eq:psi0_grid}),
can be written as 
\begin{equation}
\ket{\psi(0)} = \ket{\phi_{stat}^a} + 4a(\ket{i,j,\R} + \ket{i+1,j,\LL}),
\end{equation}
for $a=1/\sqrt{4N}$. Therefore, the only part of the initial state which is changed by the step of the algorithm is
\begin{equation}
\frac{2}{\sqrt{N}}(\ket{i,j,\R} + \ket{i+1,j,\LL}).
\end{equation}

\noindent
Let us establish an upper bound on the probability of finding a marked vertex,
\begin{equation}
p_M = \bra{\psi(t)}\left(\sum_{v\in M}\ket{v}\bra{v}\otimes I \right)\ket{\psi(t)},
\end{equation}
where $M$ is the set of marked vertices. 

\begin{lemma}
Consider a grid of size $\sqrt{N} \times \sqrt{N}$ with two adjacent marked vertices $(i,j)$ and $(i,j+1)$. Then for any number of steps, the probability of finding a marked vertex $p_M$ is $O\left(\frac{1}{N}\right)$.
\end{lemma}
\begin{proof}
Follows from the proof of Theorem~\ref{thm:2_pm} by substituting $d=4$ and $m=2N$.
%
\end{proof}

Fig.~\ref{fig:plot_grid_2_marked} shows the absolute value of the overlap, $|\braket{\psi(0)}{\psi(t)}|$, and the probability of finding a marked vertex, $p_M$, during the first 100 steps of the walk for a grid of size $50\times 50$ and two different sets of marked vertices. In the first case (solid line), we have two adjacent marked vertices, $M=\{(0,0), (0,1)\}$ and in the second case (dashed line), we have $M=\{(0,0),(0,2)\}$. Clearly, one can see the effect of the stationary state on the evolution. If the two marked vertices are adjacent, the overlap stays closes to 1 and the probability of finding a marked vertex stays close to the probability in the initial state. If the two marked vertices are not adjacent, the quantum walk behaves as expected (as in the single marked vertex case).

\begin{figure}[!htb]
\centering
\includegraphics[width=2.8in]{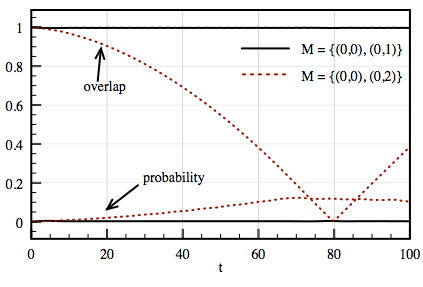}
\caption{Probability of finding a marked vertex, $p_M$, and absolute value of the overlap, $|\braket{\psi(0)}{\psi(t)}|$, for the first 100 steps of the quantum walk on a grid of size $50 \times 50$. (Solid line) We have two adjacent marked vertices, $(0,0)$ and $(0,1)$. (Dashed line) We have two non-adjacent marked vertices, $(0,0)$ and $(0,2)$.} \label{fig:plot_grid_2_marked}
\end{figure}

Note that if we have a block of marked vertices of size $k \times m$ we can construct a stationary state as long as we can tile the block by the blocks of size $1 \times 2$ and $2 \times 1$. For example, consider $M = \{(0,0),(0,1),(2,0),(3,0)\}$ for $n \geq 3$. Then the stationary state is given by
\begin{equation}
\ket{\phi_{stat}^a} = \sum_{x,y=0}^{n-1}\sum_c a\ket{x,y,c} - 4a\ket{0,0,\R} - 4a\ket{0,1,\LL} - 4a\ket{2,0,\U} - 4a\ket{3,0,\D}.
\end{equation}
More details on alternative constructions of stationary states for blocks of marked vertices on the two-dimensional grid can be found in~\cite{Nahimovs:2015}.

\section{Hypercube}\label{sec:hypercube}

\subsection{Quantum walk on the hypercube}

The $n$-dimensional hypercube is a graph with $N=2^n$ vertices where each vertex has degree $n$. The discrete-time quantum walk has associated Hilbert space ${\cal{H}} ^{2^n} \otimes {\cal{H}}^n$. The evolution operator is given by $U = S\cdot (I \otimes C)$, where the shift operator, $S$, acts in the following way
\begin{equation}
S\ket{\vec{v}}\ket{c} = \ket{\vec{v}\oplus \vec{e_c}}\ket{c},
\end{equation}
with $\vec{v}$ being the binary representation of $v$.
This operator moves the walker from state $\ket{\vec{v}}\ket{c}$ to $\ket{\vec{v}\oplus \vec{e_c}}\ket{c}$, where $\vec{e_c}$ is the binary vector with 1 in the $c$-th position. Note, that vertices are connected to each other if they differ in only one bit position. The coin transformation is the Grover coin $C = 2\ket{s}\bra{s} - I$, where $\ket{s} = \frac{1}{\sqrt{n}}\sum_{i=0}^{n-1}\ket{i}$. 

Searching for marked vertices in the hypercube is done by using the unitary operator $U'= S \cdot (I \otimes C) \cdot (Q \otimes I)$, where $Q$ is the query transformation, which flips the sign of marked vertices. The initial state of the algorithm is given by
\begin{equation}\label{eq:psi0_hyper}
\ket{\psi(0)} = \frac{1}{\sqrt{nN}}\sum_{\vec{v}=0}^{N-1}\sum_{c=0}^{n-1}\ket{\vec{v}}\ket{c}.
\end{equation}

In case of SKW algorithm with one marked vertex~\cite{Shenvi:2003}, if we measure the state of the quantum walk after $O(\sqrt{N})$ time steps, we will find the marked vertex with probability $1/2-O(1/n)$. Hence, we expect to repeat the algorithm a constant number of times, which  gives the total running time of $O(\sqrt{N})$ steps.

\subsection{Stationary states for the hypercube}

Consider a hypercube with two adjacent marked vertices $i$ and $j$. Without loss of generality, suppose $\vec{i}$ and $\vec{j}$ differ in the first bit.
Let $\ket{\phi_{stat}^a}$ be a state having amplitudes of all basis states equal to $a$ except for $\ket{\vec{i},0}$ and $\ket{\vec{j},0}$ which have amplitudes equal to $-(n-1)a$ (see Fig.~\ref{fig:2_vertex_hyper}), that is,
\begin{equation}\label{eq:phi_hyper}
\ket{\phi_{stat}^a} = a\sum_{v=0}^{N-1}\sum_{c=0}^{n-1} \ket{v,c} -an\left(\ket{\vec{i},0}+\ket{\vec{j},0}\right).
\end{equation}

\begin{figure}[!htb]
\centering
\begin{tikzpicture}[auto, thick]

\node[ccyan] (a1) at (-1,0) {\footnotesize{$i$}} ;
\node[ccyan] (a2) at (1,0) {\footnotesize{$j$}};

\draw[thick] (a1) -- node{\footnotesize{$-(n-1)a$}} (a2) ;
\draw[thick] (a1) -- node[yshift=-15pt]{\footnotesize{$-(n-1)a$}} (a2) ;
\draw[thick] (a2) -- (2,1) node[right] {\footnotesize{$\dots$}};
\draw[thick] (a2) -- (2.3,0.3) node[right] {\footnotesize{$\dots$}};
\draw[thick] (a2) -- (2,-1) node[right] {\footnotesize{$\dots$}};
\draw[thick] (a1) -- (-2,1) node[left] {\footnotesize{$\dots$}};
\draw[thick] (a1) -- (-2.3,0.3) node[left] {\footnotesize{$\dots$}};
\draw[thick] (a1) -- (-2,-1) node[left] {\footnotesize{$\dots$}};

\node at (1.4,0.7) {\footnotesize{$a$}};
\node at (1.6,0.3) {\footnotesize{$a$}};
\node at (1.3,-0.6) {\footnotesize{$a$}};
\node at (1.7,-0.2) {\footnotesize{$\vdots$}};

\node at (-1.4,0.7) {\footnotesize{$a$}};
\node at (-1.6,0.3) {\footnotesize{$a$}};
\node at (-1.3,-0.6) {\footnotesize{$a$}};
\node at (-1.7,-0.2) {\footnotesize{$\vdots$}};

\end{tikzpicture}

\caption{Amplitudes of the stationary state in an $n$-dimensional hypercube with two adjacent marked vertices $i$ and $j$.}
\label{fig:2_vertex_hyper}
\end{figure}
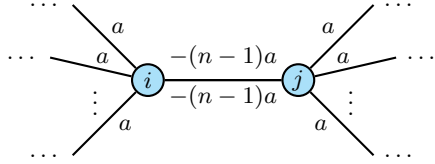

\begin{lemma}
Consider an $n$-dimensional hypercube with two adjacent marked vertices $i$ and $j$. Then $\ket{\phi_{stat}^a}$, given by Eq.~(\ref{eq:phi_hyper}), is not changed by a step of the algorithm, that is, $U'\ket{\phi_{stat}^a} = \ket{\phi_{stat}^a}$.
\end{lemma}
\begin{proof}
Similar to proof of Lemma~\ref{lemma:grid}.
\end{proof}

The initial state of the algorithm, given by Eq.~(\ref{eq:psi0_hyper}), can be written as 
\begin{equation}
\ket{\psi(0)} = \ket{\phi_{stat}^a} +an\left(\ket{\vec{i},0}+\ket{\vec{j},0}\right),
\end{equation}
for $a=1/\sqrt{nN}$. Therefore, the only part of the initial state, which is changed by a step of the algorithm is 
\begin{equation}
\sqrt{\frac{n}{N}}\left(\ket{\vec{i},0}+\ket{\vec{j},0}\right).
\end{equation}
From this fact, we can establish an upper bound for the probability of finding a marked vertex. 

\begin{lemma}
Consider an $n$-dimensional hypercube with two adjacent marked vertices $i$ and $j$. Then for any number of steps, the probability of finding a marked vertex $p_M$ is $O\left(\frac{n^2}{N}\right)$.
\end{lemma}
\begin{proof}
Follows from the proof of Theorem~\ref{thm:2_pm} by substituting $d=n$ and $m=(nN)/2$.
\end{proof}

Fig.~\ref{fig:plot_hyper_2_marked} shows the probability of finding a marked vertex and the absolute value of the overlap, $|\braket{\psi(0)}{\psi(t)}|$, for a hypercube of dimension $n=10$ for the first 100 steps of the algorithm. We consider two different sets of marked vertices. In the first case (solid line), we have two adjacent marked vertices $M=\{0,1\}$. In this case, the overlap stays close to 1 and the probability stays close to the probability in the initial state, because the quantum walk has a stationary state.
\begin{figure}[!htb]
\centering
\includegraphics[width=2.8in]{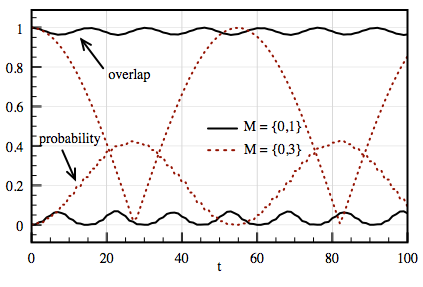}
\caption{Probability of finding a marked vertex, $p_M$, and absolute value of the overlap, $|\braket{\psi(0)}{\psi(t)}|$, for 100 steps of the quantum walk on hupercube with $N = 2^{10}$ vertices. (Solid line) We have two adjacent marked vertices, $0$ and $1$. (Dashed line) We have two non-adjacent marked vertices, $0$ and $3$.} \label{fig:plot_hyper_2_marked}
\end{figure}
In the second case (dashed line), we have two non-adjacent marked vertices $M=\{0,3\}$. 
As one can see, the behavior in the second case is very different from the behavior in the first case.

Note that any set of marked vertices which can be divided into unique blocks of two adjacent marked vertices will result in a stationary state.

\section{General graphs}\label{sec:general}

\subsection{Quantum walks on general graphs}

Consider a graph $G = (V, E)$ with a set of vertices $V$ and a set of edges $E$. Let $n = |V|$ and $m = |E|$. 
The discrete-time quantum walk on $G$ has associated Hilbert space ${\cal H}^{2m}$ with the set of basis states $\{\ket{v,c}: v \in V, 0 \leq c < d_v \}$, where $d_v$ is the degree of vertex $v$. Note, that the state $\ket{v,c}$ cannot be written as $\ket{v}\otimes\ket{c}$ unless $G$ is regular. 

The evolution operator is given by
\begin{equation}
U= SC.
\label{eq:U}
\end{equation}
The coin transformation $C$ is the direct sum of coin transformations for individual vertices, i.e. $C = C_{d_1}\bigoplus \cdots\bigoplus C_{d_n}$ with $C_{d_i}$ being the Grover diffusion transformation of dimension $d_i$.
The shift operator $S$ acts in the following way,
\begin{equation}
S\ket{v,c} = \ket{v',c'},
\label{eq:S}
\end{equation}
where $v$ and $v'$ are adjacent, $c$ and $c'$ represent the directions that points $v$ to $v'$ and $v'$ to $v$, respectively.  

Searching for marked vertices is done by using the unitary operator
\begin{equation}
U' = SCQ,
\end{equation}
where $Q$ is the query transformation, which flips the signs of the amplitudes at the marked vertices, that is,
\begin{equation}
Q = I - 2\sum_{w\in M}\sum_{c=0}^{d_w-1}\ket{w,c}\bra{w,c},
\end{equation}
with $M$ being the set of marked vertices. The initial state of the algorithm is the equal superposition over all vertex-direction pairs:
\begin{equation}\label{eq:psi0_gen}
\ket{\psi(0)} = \frac{1}{\sqrt{2m}} 
\sum_{v=0}^{n-1} \sum_{c = 0}^{d_v-1} \ket{v,c}.
\end{equation}
It can be easily verified that the initial state stays unchanged by the evolution operator $U$, regardless of the number of steps.

The running time of the algorithm depends on both the structure of the graph as well as the placement of marked vertices. Next, we describe a class of exceptional configurations of marked vertices, for which the probability of finding a marked vertex is limited.

\subsection{Stationary states for general graphs}

It is not difficult to see that the construction of stationary states for the two-dimensional grid and the hypercube can be generalized to any graph.
First, we consider only two adjacent marked vertices. Then, we show we can also construct a stationary state for three adjacent marked vertices. Later, we describe  general conditions for a state to be stationary.

\subsubsection{Two adjacent marked vertices}

Consider a graph $G = (V, E)$ with two adjacent marked vertices $i$ and $j$ with the same degree, that is, $d_i = d_j = d$.
Let $\ket{\phi_{stat}^a}$ be a state having all amplitudes equal to $a$ except of the amplitude of vertex $i$ pointing to vertex $j$ and amplitude of vertex $j$ pointing to vertex $i$, which are equal to $-(d-1)a$. Fig.~\ref{fig:2_vertex_symmetric_stationary_state} shows the configuration of the amplitudes in the marked vertices.
Then, this state is not changed by a step of the algorithm.

\begin{figure}[!htb]
\centering
\begin{tikzpicture}[auto, thick]

\node[ccyan] (a1) at (-1,0) {\footnotesize{$i$}} ;
\node[ccyan] (a2) at (1,0) {\footnotesize{$j$}};

\draw[thick] (a1) -- node{\footnotesize{$-(d-1)a$}} (a2) ;
\draw[thick] (a1) -- node[yshift=-17pt]{\footnotesize{$-(d-1)a$}} (a2) ;
\draw[thick] (a2) -- (2,1);
\draw[thick] (a2) -- (2.3,0.3);
\draw[thick] (a2) -- (2,-1);
\draw[thick] (a1) -- (-2,1);
\draw[thick] (a1) -- (-2.3,0.3);
\draw[thick] (a1) -- (-2,-1);

\node at (1.4,0.7) {\footnotesize{$a$}};
\node at (1.6,0.3) {\footnotesize{$a$}};
\node at (1.3,-0.6) {\footnotesize{$a$}};
\node at (1.7,-0.2) {\footnotesize{$\vdots$}};

\node at (-1.4,0.7) {\footnotesize{$a$}};
\node at (-1.6,0.3) {\footnotesize{$a$}};
\node at (-1.3,-0.6) {\footnotesize{$a$}};
\node at (-1.7,-0.2) {\footnotesize{$\vdots$}};

\draw [decorate,thick,black,decoration={brace,amplitude=10pt},xshift=0pt,yshift=0pt]
(2.5,1.1) -- (2.5,-1.1)  node [black,midway,xshift=9pt] {\footnotesize $d-1$};

\draw [decorate,thick,black,decoration={brace,mirror,amplitude=10pt},xshift=0pt,yshift=0pt]
(-2.5,1.1) -- (-2.5,-1.1)  node [black,midway,xshift=-35pt] {\footnotesize $d-1$};

\end{tikzpicture}

\caption{The amplitudes for the stationary state with two adjacent marked vertices $i$ and $j$.}
\label{fig:2_vertex_symmetric_stationary_state}
\end{figure}
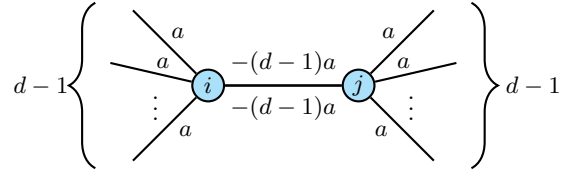

\begin{theorem}
\label{thm:2_vertex_stationary_state}
Let $G = (V,E)$ be a graph with two adjacent marked vertices $i$ and $j$ with $d_i = d_j = d$, and let
\begin{equation}
\ket{\phi_{i,j}^a} = -ad\left(\ket{i,c_{(i,j)}} + \ket{j,c_{(j,i)}}\right),
\end{equation}
where $c_{(i,j)}$ represents the direction which points vertex $i$ to vertex $j$.
Then,
\begin{equation}
\ket{\phi_{stat}^a} = a\sum_{v=0}^{n-1}\sum_{c=0}^{d_v-1} \ket{v,c} + \ket{\phi_{i,j}^a},
\end{equation}
is not affected by a step of the algorithm, that is, $U'\ket{\phi_{stat}^a}=\ket{\phi_{stat}^a}$. 
\end{theorem}

\begin{proof}
Consider the effect of a step of the algorithm. The query transformation changes the sign of all amplitudes of the marked vertices. The coin flip performs an inversion about the average of these amplitudes: for unmarked vertices, it does nothing as all amplitudes are equal to $a$; for marked vertices, the average is $0$, so it results in sign flip. Thus, $CQ$ does nothing for the amplitudes of the unmarked vertices and twice flips the sign of amplitudes of the marked vertices. Therefore, we have 
\begin{equation}
CQ\ket{\phi_{stat}^a} = \ket{\phi_{stat}^a}.
\end{equation}
The shift transformation swaps amplitudes of adjacent vertices. For $\ket{\phi_{stat}^a}$, it swaps $a$ with $a$ and $-(d-1)a$ with $-(d-1)a$. 
Thus, we have 
\begin{equation}
SCQ\ket{\phi_{stat}^a} = \ket{\phi_{stat}^a}.
\end{equation}
\end{proof}

The initial state of the algorithm $\ket{\psi(0)}$, given by Eq.~(\ref{eq:psi0_gen}), can be written as 
\begin{equation}
\ket{\psi(0)} = \ket{\phi_{stat}^a} -\ket{\phi_{i,j}^a},
\end{equation}
for $a=1/\sqrt{2m}$. Therefore, the only part of the initial state which is changed by a step of the algorithm is $\ket{\phi_{i,j}^a}$. From this fact, we can establish an upper bound for the probability of finding a marked vertex.

\begin{theorem}\label{thm:2_pm}
Let $G = (V,E)$ be a graph with two adjacent marked vertices $i$ and $j$ with $d_i = d_j = d$, and let the probability of finding a marked vertex be given by
\begin{equation}\label{eq:pm_general}
p_M = \bra{\psi(t)}\left(\sum_{v \in M}\sum_{c=0}^{d_v-1}\ket{v,c}\bra{v,c}\right)\ket{\psi(t)},
\end{equation}
where $\ket{\psi(t)} = U^t\ket{\psi(0)}$.
Then,
$p_M = O\left(\frac{d^2}{m}\right)$, 
where $m$ is the number of edges of the graph.
\end{theorem}
\begin{proof}
The only part of the initial state $\ket{\psi(0)}$ which is changed by the step of the algorithm is $\ket{\phi_{i,j}^a} = -ad\left(\ket{i,c_{(i,j)}} + \ket{j,c_{(j,i)}}\right)$, for $a = 1/\sqrt{2m}$.
Since the evolution is unitary, this part will keep its norm unchanged. In this way, we want to find how big amplitudes can get in order to maximize the value of $p_M$. This means we want to maximize the function
\begin{equation}
2(d-1)(a+x_1)^2+2(-(d-1)a-x_2)^2,
\end{equation}  
subject to $2(d-1)x_1^2+2x_2^2 = ||\ket{\phi_{i,j}^a}||^2 = 2a^2d^2$. 
Note that $x_1$ represents the amplitudes going from the marked vertices to unmarked vertices and $x_2$ represents the amplitudes going from one marked vertex to the other. Then, we obtain
\begin{equation}
p_M \leq 2a^2(2\sqrt{(d-1)d^3}+d(2d-1)) = O\left(\frac{d^2}{m}\right).
\end{equation}
\end{proof}

One of corollaries of Theorem~\ref{thm:2_pm} is that if the degree of the marked vertices is constant or if it does not grow as a function of $n$, then for large $n$, the probability of finding a marked vertex will stay close to the initial probability.

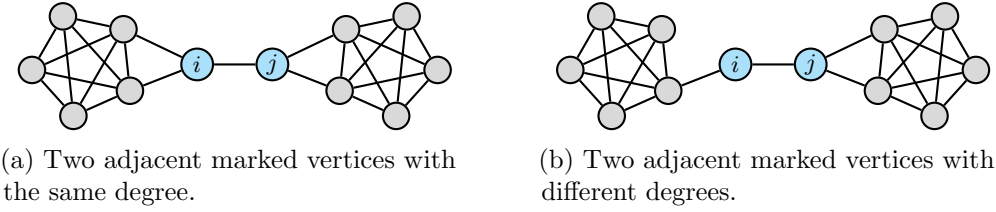
\begin{figure}[!htb]
\centering
\subcaptionbox{Two adjacent marked vertices with the same degree.\label{fig:2_marked_example_a}}{
\centering
\begin{tikzpicture}[auto, thick]
\node[ccyan] (a1) at (-0.5,0) {\footnotesize{$i$}} ;
\node[ccyan] (a2) at (0.5,0) {\footnotesize{$j$}};
\begin{scope}[scale=0.7,rotate=-48]
\foreach \name/\angle/\text in {P-4/234/4, P-5/162/5, 
                                  P-1/90/1, P-2/18/2, P-3/-54/3}
    \node[cgrey,xshift=-2cm,yshift=0cm] (\name) at (\angle:1cm) {};
\foreach \from/\to in {1/2,2/3,3/4,4/5,5/1,1/3,2/4,3/5,4/1,5/2}
    { \draw (P-\from) -- (P-\to);}
\end{scope}
\begin{scope}[scale=0.7,rotate=48]
\foreach \name/\angle/\text in {R-4/234/4, R-5/162/5, 
                                  R-1/90/1, R-2/18/2, R-3/-54/3}
    \node[cgrey,xshift=2cm,yshift=0cm] (\name) at (\angle:1cm) {};
\foreach \from/\to in {1/2,2/3,3/4,4/5,5/1,1/3,2/4,3/5,4/1,5/2}
    {\draw (R-\from) -- (R-\to);}
\end{scope}
\draw[thick] (a1) --  (a2) ;
\draw[thick] (a2) --  (R-5) ;
\draw[thick] (a2) --  (R-1) ;
\draw[thick] (a1) --  (P-2) ;
\draw[thick] (a1) --  (P-1) ;
\end{tikzpicture}
}\hspace{1cm}
\subcaptionbox{Two adjacent marked vertices with different degrees.\label{fig:2_marked_example_b}}{
\centering
\begin{tikzpicture}[auto, thick]
\node[ccyan] (a1) at (-0.5,0) {\footnotesize{$i$}} ;
\node[ccyan] (a2) at (0.5,0) {\footnotesize{$j$}};
\begin{scope}[scale=0.7,rotate=-48]
\foreach \name/\angle/\text in {P-4/234/4, P-5/162/5, 
                                  P-1/90/1, P-2/18/2, P-3/-54/3}
    \node[cgrey,xshift=-2cm,yshift=0cm] (\name) at (\angle:1cm) {};
\foreach \from/\to in {1/2,2/3,3/4,4/5,5/1,1/3,2/4,3/5,4/1,5/2}
    { \draw (P-\from) -- (P-\to);}
\end{scope}
\begin{scope}[scale=0.7,rotate=48]
\foreach \name/\angle/\text in {R-4/234/4, R-5/162/5, 
                                  R-1/90/1, R-2/18/2, R-3/-54/3}
    \node[cgrey,xshift=2cm,yshift=0cm] (\name) at (\angle:1cm) {};
\foreach \from/\to in {1/2,2/3,3/4,4/5,5/1,1/3,2/4,3/5,4/1,5/2}
    { \draw (R-\from) -- (R-\to);}
\end{scope}
\draw[thick] (a1) --  (a2) ;
\draw[thick] (a2) --  (R-5) ;
\draw[thick] (a2) --  (R-1) ;
\draw[thick] (a1) --  (P-2) ;
\end{tikzpicture}
}
\caption{Graphs with two adjacent marked vertices $i$ and $j$ connected to two complete graphs of size $5$. $(a)$ The marked vertices have the same degree. $(b)$ The marked vertices have different degrees.}
\label{fig:2_marked_example}
\end{figure}

Consider the example of the graphs in Fig.~\ref{fig:2_marked_example}.
Intuitively, one might expect that the probability of finding a marked vertex in graph \ref{fig:2_marked_example_a} would be bigger than in \ref{fig:2_marked_example_b}. However, Fig.~\ref{fig:plot_2_marked} shows something different. We calculate the probability of finding a marked vertex and the absolute value of the overlap, $|\braket{\psi(0)}{\psi(t)}|$, for 100 steps of the evolution and when the size of the two complete graphs attached to the marked vertices is 10. The graph \ref{fig:2_marked_example_a} produces a stationary state, and that is why its overlap stays very close to 1 and the probability of finding a marked vertex is smaller than in the other graph.
\begin{figure}[!htb]
\centering
\includegraphics[width=2.8in]{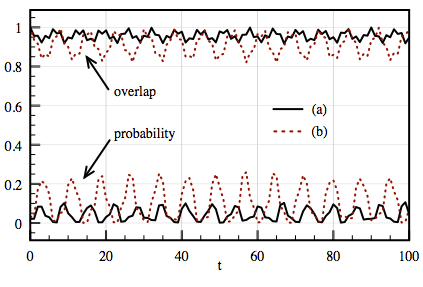}
\caption{Probability of finding a marked vertex, $p_M$, and absolute value of the overlap, $|\braket{\psi(0)}{\psi(t)}|$, for 100 steps of the quantum walk. (Solid line) Graph~\ref{fig:2_marked_example_a} with two adjacent marked vertices with the same degree and complete graphs of size $10$. (Dashed line) Graph~\ref{fig:2_marked_example_b} with two adjacent marked vertices with different degree and complete graphs of size $10$.} \label{fig:plot_2_marked}
\end{figure}

\begin{corollary}
Suppose the set of marked vertices, $M$, can be divided in groups of two adjacent marked vertices without intersecting each other. Let $M'$ be the set of such pairs. Then
\begin{equation}
\ket{\phi_{stat}^a} = \sum_{v=0}^{n-1} \sum_{c=0}^{d_v-1} a\ket{v,c} +\sum_{(i,j)\in M'} \ket{\phi_{i,j}^a}
\end{equation}
is not changed by a step of the algorithm, that is, $U'\ket{\phi_{stat}^a} = \ket{\phi_{stat}^a}$.
\end{corollary}
\begin{proof}
Follows from the proof of Theorem~\ref{thm:2_vertex_stationary_state}.
\end{proof}

\subsubsection{Three adjacent marked vertices}

Now, consider a graph $G=(V,E)$ with three adjacent marked vertices $i$, $j$ and $k$, that is, a marked triangle. The stationary state for this case will have the amplitudes in the marked vertices as depicted in Fig.~\ref{fig:3_marked}.

\begin{figure}[!htb]
\centering
\begin{tikzpicture}[auto, thick]

\node[ccyan] (a1) at (-1,0) {\footnotesize{$i$}} ;
\node[ccyan] (a2) at (1,0) {\footnotesize{$j$}};
\node[ccyan] (a3) at (0,1.5) {\footnotesize{$k$}};

\draw[thick] (a1) -- node [midway,yshift=-15pt]{\footnotesize{$-al_{ij}$}} (a2) ;
\draw[thick] (a1) -- node [midway,xshift=4pt]{\footnotesize{$-al_{ik}$}} (a3) ;
\draw[thick] (a3) -- node [midway,xshift=-5pt] {\footnotesize{$-al_{jk}$}} (a2) ;

\draw[thick] (a3) -- (0.8,2.1);
\draw[thick] (a3) -- (0.2,2.5);
\draw[thick] (a3) -- (-0.8,2.2);

\draw[thick] (a2) -- (1.9,0.3);
\draw[thick] (a2) -- (1.8,-0.5);
\draw[thick] (a2) -- (1.2,-0.9);

\draw[thick] (a1) -- (-1.9,0.3);
\draw[thick] (a1) -- (-1.8,-0.5);
\draw[thick] (a1) -- (-1.2,-0.9);

\node at (1.4,0.3) {\footnotesize{$a$}};
\node at (1.4,-0.4) {\footnotesize{$a$}};
\node at (1,-0.5) {\footnotesize{$a$}};
\node at (1.6,0) {\footnotesize{$\vdots$}};

\node at (-1.4,0.3) {\footnotesize{$a$}};
\node at (-1.4,-0.4) {\footnotesize{$a$}};
\node at (-1,-0.5) {\footnotesize{$a$}};
\node at (-1.6,0) {\footnotesize{$\vdots$}};

\node at (0.5,1.7) {\footnotesize{$a$}};
\node at (0.25,2.1) {\footnotesize{$a$}};
\node at (-0.5,1.75) {\footnotesize{$a$}};
\node at (-0.2,2) {\footnotesize{$\cdots$}};

\draw [decorate,thick,black,decoration={brace,amplitude=10pt},xshift=0pt,yshift=0pt,rotate=-15]
(2,0.8) -- (1.6,-0.8)  node [black,midway,xshift=9pt] {\footnotesize $d_j-2$};

\draw [decorate,thick,black,decoration={brace,mirror,amplitude=10pt},xshift=0pt,yshift=0pt]
(0.9,2.4) -- (-0.9,2.4)  node [black,midway,yshift=23pt] {\footnotesize $d_k-2$};

\draw [decorate,thick,black,decoration={brace,mirror,amplitude=10pt},xshift=0pt,yshift=0pt,rotate=15]
(-2,0.8) -- (-1.6,-0.8)  node [black,midway,xshift=-40pt,yshift=-15pt] {\footnotesize $d_i-2$};

\end{tikzpicture}
\caption{Sketch of amplitudes for the stationary state with three adjacent marked vertices $i$, $j$ and $k$.}
\label{fig:3_marked}
\end{figure}
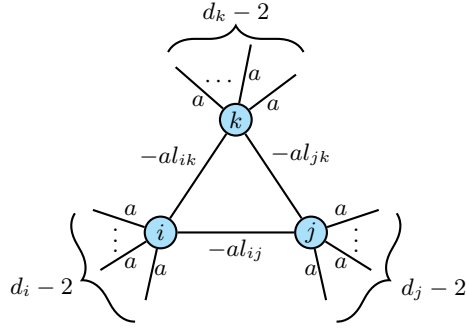

Note that in order to have a stationary state the sum of amplitudes of each marked vertex should be 0, so the action of the coin operator will be a sign flip. 
By solving the following system of equations:
\begin{equation}
\label{eq:3_vertex_generic_stationary_state}
\begin{cases} 
l_{ij} + l_{ik} = d_i-2 \\ 
l_{ij} + l_{jk} = d_j-2 \\
l_{ik} + l_{jk} = d_k-2 \\
\end{cases} 
\end{equation}
we obtain,
\begin{equation}\label{eq:ls}
l_{ij} = \frac{d_i+d_j-d_k}{2}-1,\quad l_{ik} = \frac{d_i+d_k-d_j}{2}-1 \quad{\textrm{and}}\quad l_{jk} = \frac{d_j+d_k-d_i}{2}-1.
\end{equation}

\begin{theorem}
\label{thm:3_vertex_stationary_state}
Let $G = (V,E)$ be a graph with three adjacent marked vertices $i$, $j$ and $k$; and let
\begin{equation}
\begin{split}
\ket{\phi_{i,j,k}^a} =& -a(l_{ij}+1)\left(\ket{i,c_{(i,j)}} + \ket{j,c_{(j,i)}}\right) -a(l_{ik}+1)\left(\ket{i,c_{(i,k)}} + \ket{k,c_{(k,i)}}\right)-\\
&-a(l_{jk}+1)\left(\ket{i,c_{(j,k)}} + \ket{k,c_{(k,j)}}\right),
\end{split}
\end{equation}
where $l_{ij}, l_{ik}$, and $l_{jk}$ are defined in (\ref{eq:ls}). Then,
\begin{equation}
\ket{\phi_{stat}^a} = a\sum_{v=0}^{n-1} \sum_{c=0}^{d_v-1} \ket{v,c} + \ket{\phi_{i,j,k}^a},
\end{equation}
is not affected by a step of the quantum walk on $G$. 
\end{theorem}

\begin{proof}
Similar to Theorem~\ref{thm:2_vertex_stationary_state}.
\end{proof}

Fig.~\ref{fig:3_marked_example} shows two graphs with three adjacent marked vertices $i,j$ and $k$ connected to three complete graphs. Remember, that for three adjacent marked vertices it doesn't matter if their degrees are different or equal. Both cases will result in a stationary state.

\begin{figure}[!htb]
\centering
\subcaptionbox{Three adjacent marked vertices with $d_i=4, d_j=3$ and $d_k = 5$.\label{fig:3_marked_example_a}}{
\centering
\begin{tikzpicture}[auto, thick]

\node[ccyan] (a1) at (-0.5,0) {\footnotesize{$i$}} ;
\node[ccyan] (a2) at (0.5,0) {\footnotesize{$j$}};
\node[ccyan] (a3) at (0,1) {\footnotesize{$k$}};

\draw[thick] (a1) -- (a2) ;
\draw[thick] (a1) -- (a3) ;
\draw[thick] (a3) -- (a2) ;

\begin{scope}[scale=0.7,rotate=-58]
\foreach \name/\angle/\text in {P-4/234/4, P-5/162/5, 
                                  P-1/90/1, P-2/18/2, P-3/-54/3}
    \node[cgrey,xshift=-2cm,yshift=-0.8cm] (\name) at (\angle:1cm) {};
\foreach \from/\to in {1/2,2/3,3/4,4/5,5/1,1/3,2/4,3/5,4/1,5/2}
    { \draw (P-\from) -- (P-\to);}
\end{scope}

\begin{scope}[scale=0.7,rotate=58]
\foreach \name/\angle/\text in {R-4/234/4, R-5/162/5, 
                                  R-1/90/1, R-2/18/2, R-3/-54/3}
    \node[cgrey,xshift=2cm,yshift=-0.8cm] (\name) at (\angle:1cm) {};
\foreach \from/\to in {1/2,2/3,3/4,4/5,5/1,1/3,2/4,3/5,4/1,5/2}
    { \draw (R-\from) -- (R-\to);}
\end{scope}

\begin{scope}[scale=0.7,rotate=180]
\foreach \name/\angle/\text in {T-4/234/4, T-5/162/5, 
                                  T-1/90/1, T-2/18/2, T-3/-54/3}
    \node[cgrey,xshift=0cm,yshift=2.5cm] (\name) at (\angle:1cm) {};
\foreach \from/\to in {1/2,2/3,3/4,4/5,5/1,1/3,2/4,3/5,4/1,5/2}
    { \draw (T-\from) -- (T-\to);}
\end{scope}

\draw[thick] (a1) -- (P-1) ;
\draw[thick] (a1) -- (P-2) ;
\draw[thick] (a2) -- (R-1) ;
\draw[thick] (a3) -- (T-1) ;
\draw[thick] (a3) -- (T-2) ;
\draw[thick] (a3) -- (T-5) ;

\end{tikzpicture}}
\hspace{1cm}
\subcaptionbox{Three adjacent marked vertices with $d_i=4, d_j=3$ and $d_k = 6$.\label{fig:3_marked_example_b}}{
\centering
\begin{tikzpicture}[auto, thick]

\begin{scope}[scale=0.7,rotate=-58]
\foreach \name/\angle/\text in {P-4/234/4, P-5/162/5, 
                                  P-1/90/1, P-2/18/2, P-3/-54/3}
    \node[cgrey,xshift=-2cm,yshift=-0.8cm] (\name) at (\angle:1cm) {};
\foreach \from/\to in {1/2,2/3,3/4,4/5,5/1,1/3,2/4,3/5,4/1,5/2}
    { \draw (P-\from) -- (P-\to);}
\end{scope}

\begin{scope}[scale=0.7,rotate=58]
\foreach \name/\angle/\text in {R-4/234/4, R-5/162/5, 
                                  R-1/90/1, R-2/18/2, R-3/-54/3}
    \node[cgrey,xshift=2cm,yshift=-0.8cm] (\name) at (\angle:1cm) {};
\foreach \from/\to in {1/2,2/3,3/4,4/5,5/1,1/3,2/4,3/5,4/1,5/2}
    { \draw (R-\from) -- (R-\to);}
\end{scope}

\begin{scope}[scale=0.7,rotate=180,yshift=1.1cm]
\foreach \name/\angle/\text in {T-4/234/4, T-5/162/5, 
                                  T-2/18/2, T-3/-54/3}
    \node[cgrey,xshift=0cm,yshift=2.5cm] (\name) at (\angle:1cm) {};
\node[ccyan,xshift=0cm,yshift=2.5cm] (T-1) at (90:1cm) {\footnotesize{$k$}};
\foreach \from/\to in {1/2,2/3,3/4,4/5,5/1,1/3,2/4,3/5,4/1,5/2}
    { \draw (T-\from) -- (T-\to);}
\end{scope}

\node[ccyan] (a1) at (-0.5,0) {\footnotesize{$i$}} ;
\node[ccyan] (a2) at (0.5,0) {\footnotesize{$j$}};

\draw[thick] (a1) -- (a2) ;
\draw[thick] (a1) -- (T-1) ;
\draw[thick] (T-1) -- (a2) ;

\draw[thick] (a1) -- (P-1) ;
\draw[thick] (a1) -- (P-2) ;
\draw[thick] (a2) -- (R-1) ;

\end{tikzpicture}}
\caption{Graphs with three marked vertices connected to three complete graphs of size $5$.}
\label{fig:3_marked_example}
\end{figure}
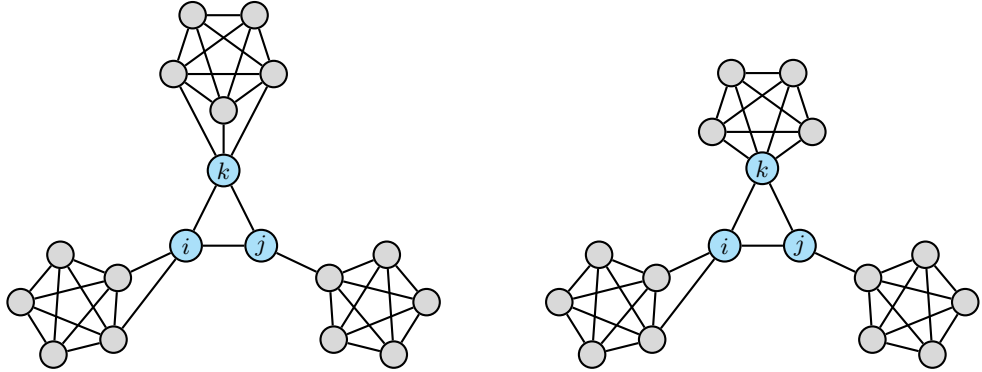

In Fig.~\ref{fig:plot_3_marked}, we can see how the probability of finding a marked vertex depends on the degree of the marked vertices. Although, both graphs have a stationary state, the maximum probability achieved by \ref{fig:3_marked_example_b} is bigger. Moreover, if we consider that the size of the complete graphs can grow, then for graph \ref{fig:3_marked_example_a} the probability of finding a marked vertex will stay close to the initial probability (the degree of the marked vertices does not scales as a function of the number of vertices), while for graph \ref{fig:3_marked_example_b} it will increase ($d_k$ increases with the size of the complete graph). 
\begin{figure}[!htb]
\centering
\includegraphics[width=2.5in]{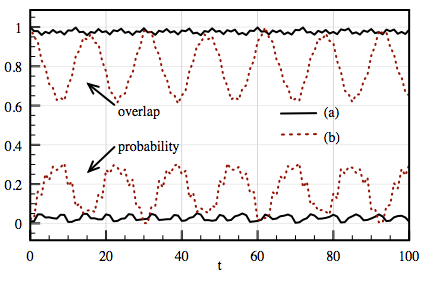}
\caption{Probability of finding a marked vertex, $p_M$, and absolute value of the overlap, $|\braket{\psi(0)}{\psi(t)}|$, for 100 steps of the quantum walk. (Solid line) Graph~\ref{fig:3_marked_example_a} with  complete graphs of size $20$ and three adjacent marked vertices with $d_i = 4, d_j=3$ and $d_k = 5$. (Dashed line) Graph~\ref{fig:3_marked_example_b} with  complete graphs of size $20$ and three adjacent marked vertices with $d_i = 4, d_j=3$ and $d_k = 21$.} \label{fig:plot_3_marked}
\end{figure}

\subsubsection{$k$-clique of marked vertices}
Next, we generalize the previous result for any complete subgraph of marked vertices.

\begin{theorem}
Let $G = (V,E)$ be a graph with $k$ marked vertices $v_{1},\dots,v_{k}$ forming a $k$-clique. Then it forms an exceptional configuration.
\end{theorem}

\begin{proof}
Let $d_{v_j} = (k-1) + d'_j$, where $d'_j$ is the number of edges of $v_{j}$ outside the clique. To construct a stationary state, we need to assign amplitudes to internal edges of the clique, so that the amplitudes in vertex $v_{j}$ sum up to $d'_j$.

Without a loss of generality let $d'_1 < d'_2 < \dots < d'_k$. 
We set the amplitude of the edge $(v_{1}, v_{2})$ to $-ad'_1$ and amplitudes of other edges within the clique outgoing from $v_{1}$ to $0$. By this, we have satisfied the condition for the vertex $v_{1}$ and reduced the problem from size $k$ to $k-1$.
 I.e. now we have a $(k-1)$-clique with degrees $(d'_2 - d'_1), d'_3, \dots, d'_k$. Next, we recursively repeat the previous step until we get a $3$-clique, which always have an assignment. In this way, we have constructed a stationary state for a $k$-clique of marked vertices.
 
\end{proof}

Note that any set of marked vertices which can be divided into unique blocks of two adjacent marked vertices with the same degree and/or $k$-clique marked vertices will result in a stationary state.


\subsubsection{General Conditions}

In this section, we describe general conditions in which a state is stationary.

\begin{theorem}[General conditions]\label{thm:general_conditions}
Let $\ket{\psi}$ be a state with the following properties:
\begin{itemize}
\item[1] All amplitudes of the unmarked vertices are equal;
\item[2] The sum of the amplitudes of any marked vertex is 0; 
\item[3] The amplitudes of two adjacent vertices pointing to each other are equal.
\end{itemize}
Then, $\ket{\psi}$ is not changed by a step of the quantum walk, that is, $U\ket{\psi} = \ket{\psi}$.
\end{theorem}
\begin{proof}
Item $1$ is required in order for the coin transformation to have no effect on the unmarked vertices. It is easy to see that $C_n\ket{u} = \ket{u}$, where $\ket{u} = a\sum_{i=0}^{n-1}\ket{i}$ for some constant $a$.

Item $2$ is necessary so the coin transformation can flip the signs of the amplitudes in the marked vertices. Note that previously the sign of these amplitudes were inverted by the query transformation.

Item $3$ is necessary for the shift transformation to have no effect on the state.
\end{proof}

Note that the aforementioned conditions are established for the case $CQ\ket{\psi} = \ket{\psi}$ and $S\ket{\psi} = \ket{\psi}$.
There might be even more general conditions for the case $U'\ket{\psi} = \ket{\psi}$.


\section{Conclusions}\label{sec:conclusions}

In this paper, we have demonstrated a wide class of exceptional configurations of marked vertices for the quantum walk based search on various graphs. The above phenomenon is purely quantum. Classically, additional marked vertices result in the decrease of the number of steps of the algorithm and the increase of the probability of finding a marked location.
Quantumly, the addition of a marked vertex can drastically drop the probability of finding a marked location.

An open question is whether the found phenomenon has analogs for other models of quantum walks (continuous time quantum walks~\cite{Farhi:1998}, Szegedy's quantum walk~\cite{Szegedy:2004}, staggered quantum walks~\cite{Portugal:2015}, etc.) as well as for alternative coin operators.

Another open question is algorithmic applications of the found effect.
For example, in case of two-dimensional grid the search algorithm can ``distinguish'' between odd-times-odd and even-times-even groups of marked locations. Moreover, if there are multiple odd-times-odd and even-times-even groups of marked locations, the algorithm will find only odd-times-odd groups and will ``ignore'' even-times-even groups.
Nothing like this is possible for classical random walks without adding additional memory resources and complicating the algorithm.
Another example is the general graphs where the algorithm ``ignores'' marked triangles.
We think that the found phenomenon should have algorithmic applications which would be very interesting to find.


\subparagraph*{Acknowledgements.}

The authors thank Andris Ambainis for helpful ideas and suggestions, which was a great help during this research, and Tom Wong for useful comments and careful revision of the manuscript.


\bibliography{Paper}


\end{document}